\documentclass{amsart}
\usepackage{hyperref}

\newcommand{\arxiv}[1]{\href{http://arxiv.org/abs/#1}{arXiv:#1}}
\newcommand*{\mailto}[1]{\href{mailto:#1}{\nolinkurl{#1}}}

\newtheorem{theorem}{Theorem}[section]
\newtheorem{lemma}[theorem]{Lemma}

\newtheorem{remark}[theorem]{Remark}
\newtheorem{hypo}[theorem]{Hypothesis {\bf H.}\hspace*{-0.6ex}}

\newcommand{\R}{{\mathbb R}}
\newcommand{\N}{{\mathbb N}}
\newcommand{\Z}{{\mathbb Z}}
\newcommand{\C}{{\mathbb C}}

\newcommand{\nn}{\nonumber}
\newcommand{\be}{\begin{equation}}
\newcommand{\ee}{\end{equation}}

\newcommand{\ol}{\overline}

\newcommand{\spr}[2]{\langle #1 , #2 \rangle}

\newcommand{\I}{\mathrm{i}}
\newcommand{\E}{\mathrm{e}}
\newcommand{\re}{\mathrm{Re}}

\newcommand{\lz}{\ell^2(\Z)}
\newcommand{\tl}{\mathrm{TL}}

\newcommand{\Ker}{\mathrm{Ker}}
\newcommand{\clos}{\mathop{\rm clos}}

\newcommand{\floor}[1]{\lfloor#1 \rfloor}

\newcommand{\sig}{\sigma}
\newcommand{\lam}{\lambda}

\numberwithin{equation}{section}

\begin{document}

\title[Inverse Scattering Transform for the Toda Hierarchy]{Inverse Scattering 
Transform for the Toda Hierarchy with Steplike Finite-Gap Backgrounds}

\author[I. Egorova]{Iryna Egorova}
\address{B. Verkin Institute for Low Temperature Physics\\ 47 Lenin Ave.\\ 61164 Kharkiv\\ Ukraine}
\email{\mailto{iraegorova@gmail.com}}

\author[J. Michor]{Johanna Michor}
\address{Faculty of Mathematics\\
Nordbergstrasse 15\\ 1090 Wien\\ Austria -- and -- International Erwin Schr\"odinger
Institute for Mathematical Physics\\ Boltzmanngasse 9\\ 1090 Wien\\ Austria}
\email{\mailto{Johanna.Michor@esi.ac.at}}
\urladdr{\url{http://www.mat.univie.ac.at/~jmichor/}}

\author[G. Teschl]{Gerald Teschl}
\address{Faculty of Mathematics\\
Nordbergstrasse 15\\ 1090 Wien\\ Austria -- and -- International Erwin Schr\"odinger
Institute for Mathematical Physics\\ Boltzmanngasse 9\\ 1090 Wien\\ Austria}
\email{\mailto{Gerald.Teschl@univie.ac.at}}
\urladdr{\url{http://www.mat.univie.ac.at/~gerald/}}

\thanks{Research supported by the Austrian Science Fund (FWF) under Grants No.\ Y330, V120.}
\thanks{J. Math. Phys. {\bf 50}, 103521 (2009)}


\begin{abstract}
We provide a rigorous treatment of the inverse scattering transform for
the entire Toda hierarchy for solutions which are asymptotically close
to (in general) different finite-gap solutions as $n\to\pm\infty$. 
\end{abstract}

\maketitle

\section{Introduction}

The Toda lattice is one of the most prominent discrete integrable wave equations.
In particular, it can be solved via the inverse scattering method. For the classical
case, where the solution is asymptotically equal to the (same) constant
solution, this is of course well understood and covered in several monographs (e.g.)
\cite{fad}, \cite{ta}, or \cite{tjac}. The corresponding long-time asymptotics were first
computed by Novokshenov and Habibullin \cite{nh} and
were later made rigorous by Kamvissis \cite{km} under the additional assumption that no
solitons are present (the case of solitons was recently added in \cite{krt}; see also
the review \cite{krtrhp}).

The inverse scattering transform for the entire Toda hierarchy in the case of a finite-gap
background was solved only recently by us in \cite{emtist} as a continuation of \cite{tist}.
Similar results were obtained by Khanmamedov \cite{kha2}.
Long-time asymptotics for such solutions have been given by Kamvissis and Teschl
for the case without solitons \cite{kt}, \cite{kt2}, \cite{kt3} and by Kr\"uger and Teschl for the
case with solitons \cite{krt2} (for related trace formulas and conserved quantities see \cite{mtqptr}).

In this respect it is important to mention that even the important case of a one-soliton solution
on a finite-gap background has different spatial asymptotics as $n\to\pm\infty$ and
hence is not covered by the above results (see \cite{emt4}, \cite{tag}). Hence this clearly
raises the need to extend the results from \cite{emtist} to the case of solutions which
are asymptotically equal to (in general) different finite-gap solutions as $n\to\pm\infty$.

In fact, the simplest case, where the solution is asymptotically equal to two different constant
solutions, has already attracted considerable interest in the past.
The first to solve the corresponding Cauchy problem (in the case of rapid
decay with respect to the background) seems to be Oba \cite{oba}. Moreover,
the long-time asymptotics were considered in \cite{bdme2}, \cite{dkv}, \cite{dkkz}, \cite{gk},
\cite{hfm}, \cite{km2}, \cite{vdo}.

Our aim here is to fill this gap and to provide a treatment of the inverse scattering transform
for the entire Toda hierarchy in the case of steplike quasi-periodic finite-gap backgrounds.
Note that since we treat the entire Toda hierarchy, our results also cover the Kac--van Moerbeke
hierarchy as a special case \cite{mtlax}.

Finally, we remark that the corresponding result for the Korteweg--de Vries equation is
much more involved and was only recently solved by Egorova, Grunert, and Teschl \cite{egt}
under some additional restrictions on the spectra of the background operators.

After introducing the Toda hierarchy in Section~\ref{secth}, we will first show that
a solution will stay close to a given background solution in Section~\ref{secivp}.
This result implies that a short-range perturbation of a steplike finite-gap
solution will stay short-range for all time, and it shows that the time-dependent scattering
data satisfy the hypothesis necessary for the Gel'fand--Levitan--Marchenko theory \cite{mar}.
This result constitutes the main technical ingredient for the inverse scattering
transform. In Section~\ref{secqp} we review some necessary facts on
quasi-periodic finite-gap solutions and in Section~\ref{secist} we compute
the time dependence of the scattering data and discuss its dynamics.

\section{The Toda hierarchy}
\label{secth}

In this section we introduce the Toda hierarchy using the standard Lax formalism
(\cite{lax}). We first review some basic facts from \cite{bght} (see also \cite{ghmt}, \cite{tjac}).

We will only consider bounded solutions and hence require

\begin{hypo} \label{H t0}
Suppose $a(t)$, $b(t)$ satisfy
\[
a(t) \in \ell^{\infty}(\Z, \R), \qquad b(t) \in \ell^{\infty}(\Z, \R), \qquad
a(n,t) \neq 0, \qquad (n,t) \in \Z \times \R,
\]
and let $t \mapsto (a(t), b(t))$ be differentiable in 
$\ell^{\infty}(\Z) \oplus \ell^{\infty}(\Z)$.
\end{hypo}

\noindent
Associated with $a(t), b(t)$ is a Jacobi operator
\begin{equation}
H(t): \lz  \to  \lz, \qquad f \mapsto \tau(t) f,
\end{equation}
where
\begin{equation}
\tau(t) f(n)= a(n,t) f(n+1) + a(n-1,t) f(n-1) + b(n,t) f(n)
\end{equation}
and $\lz$ denotes the Hilbert space of square summable (complex-valued) sequences
over $\Z$. Moreover, choose constants $c_0=1$, $c_j$, $1\le j \le r$, $c_{r+1}=0$, set
\begin{align} \nn
g_j(n,t) &= \sum_{\ell=0}^j c_{j-\ell} \spr{\delta_n}{H(t)^\ell \delta_n},\\ \label{todaghsp}
h_j(n,t) &= 2 a(n,t) \sum_{\ell=0}^j c_{j-\ell}  \spr{\delta_{n+1}}{H(t)^\ell
\delta_n} + c_{j+1},
\end{align}
where $\spr{\delta_m}{A \delta_n}$ denote the matrix elements of an operator $A$ with
respect to the standard basis, and consider the Lax operator
\begin{equation}  \label{btgptdef}
P_{2r+2}(t) = -H(t)^{r+1} + \sum_{j=0}^r ( 2a(t) g_j(t) S^+ -h_j(t)) H(t)^{r-j} +
g_{r+1}(t),
\end{equation}
where $S^\pm f(n) = f(n\pm1)$. Restricting to the two-dimensional nullspace 
\[
\Ker(\tau(t)-z), \quad z\in\C,
\]
of $\tau(t)-z$, we have the following representation of $P_{2r+2}(t)$:
\begin{equation} \label{btqptFG}
P_{2r+2}(t)\Big|_{\Ker(\tau(t)-z)} =2a(t) G_r(z,.,t) S^+ - H_{r+1}(z,.,t),
\end{equation}
where $G_r(z,n,t)$ and $H_{r+1}(z,n,t)$ are monic polynomials in $z$ of the
type
\begin{align} \nn
G_r(z,n,t) &= \sum_{j=0}^r z^j g_{r-j}(n,t),\\ \label{fgdef}
H_{r+1}(z,n,t) &= z^{r+1} + \sum_{j=0}^r z^j h_{r-j}(n,t) - g_{r+1}(n,t).
\end{align}
A straightforward computation shows that the Lax equation
\begin{equation} \label{laxp}
\frac{d}{dt} H(t) -[P_{2r+2}(t), H(t)]=0, \qquad t\in\R,
\end{equation}
is equivalent to
\begin{align} \nn
\tl_r (a(t), b(t))_1 &= \dot{a}(n,t) -a(n,t) \Big(g_{r+1}(n+1,t) -
g_{r+1}(n,t) \Big)=0,\\ \label{tlrabo}
\tl_r (a(t), b(t))_2 &= \dot{b}(n,t) - \Big(h_{r+1}(n,t) -h_{r+1}(n-1,t) \Big)=0,
\end{align}
where the dot denotes a derivative with respect to $t$.
Varying $r\in \N_0$ yields the Toda hierarchy
$\tl_r(a,b) =(\tl_r (a,b)_1, \tl_r (a,b)_2) =0$. We will always consider $r$ as
a fixed, but arbitrary, value.

Finally, we recall that the Lax equation \eqref{laxp} implies existence of
a unitary propagator $U_r(t,s)$ such that the family of operators
$H(t)$, $t\in\R$, are unitarily equivalent, $H(t) = U_r(t,s) H(s) U_r(s,t)$.

\section{The initial value problem}
\label{secivp}

First of all we recall the basic existence and uniqueness theorem for the Toda hierarchy
(see, e.g., \cite{tist}, \cite{ttkm}, or \cite[Section~12.2]{tjac}).

\begin{theorem} \label{thmexistandunique}
Suppose $(a_0,b_0) \in M = \ell^\infty(\Z) \oplus \ell^\infty(\Z)$. Then there exists a unique
integral curve $t \mapsto (a(t),b(t))$ in $C^\infty(\R,M)$ of the Toda hierarchy, that is,
$\tl_r(a(t),b(t))=0$, such that $(a(0),b(0)) = (a_0,b_0)$.
\end{theorem}

\noindent
In \cite{tist} it was shown that solutions which are asymptotically close to the constant
solution at the initial time stay close for all time. Our first aim is to extend this
result to include perturbations of quasi-periodic finite-gap solutions. In fact,
we will even be a bit more general. Set
\be
\|(a,b)\|_{w,p} = \begin{cases}
\left(\sum\limits_{n \in \Z} w(n) \Big(|a(n)|^p + |b(n)|^p \Big)\right)^{1/p}, & 1\le p <\infty,\\
\sup\limits_{n \in \Z} w(n) \Big(|a(n)| + |b(n)| \Big), & p=\infty.
\end{cases}
\ee
Then

\begin{lemma} \label{lemtsr}
Let $w(n)\ge 1$ be some weight with $\sup_n( |\frac{w(n+1)}{w(n)}| + |\frac{w(n)}{w(n+1)}| )<\infty$ and
fix some $1 \le p \le \infty$.
Suppose $a(n,t)$, $b(n,t)$ and $a_\pm(n,t)$, $b_\pm(n,t)$ are arbitrary 
bounded solutions of the Toda hierarchy and abbreviate
\be
\bar a(n,t)= \begin{cases} a_+(n,t), & n\ge 0,\\ a_-(n,t), & n<0,\end{cases}\qquad
\bar b(n,t)= \begin{cases} b_+(n,t), & n\ge 0,\\ b_-(n,t), & n<0.\end{cases}
\ee
Then, if
\be \label{H t2}
\|\big(a(t)-\bar a(t), b(t)-\bar b(t)\big)\|_{w,p} < \infty
\ee
holds for one $t=t_0 \in \R$, then it holds for all $t \in \R$.
\end{lemma}

\begin{proof}
Without loss of generality we assume that $t_0 = 0$. Let us consider the differential
equation for the differences $\delta(n,t)= \big(a(n,t) - \bar a(n,t), b(n,t) - \bar b(n,t) \big)$
in the Banach space of pairs of bounded sequences $\delta=(\delta_1,\delta_2)$ for which
the norm $\| \delta \|_{w,p}$ is finite.
We claim that $\delta$ satisfies an inhomogeneous linear differential equation of the form
\[
\dot{\delta}(t) = \sum_{|j|\le r+1} A_{r,j}(t) (S^+)^j \delta(t) + B_r(t)
\]
(see e.g.\ \cite{dei} for the theory of ordinary differential equations in Banach spaces).
Here $S^\pm(\delta_1(n,t),\delta_2(n,t))=(\delta_1(n\pm 1,t),\delta_2(n\pm 1,t))$ are the shift operators,
\[
A_{r,j}(n,t)= \begin{pmatrix}
A_{r,j}^{11}(n,t) & A_{r,j}^{12}(n,t)\\ A_{r,j}^{21}(n,t) & A_{r,j}^{22}(n,t)
\end{pmatrix},
\]
are multiplication operators with bounded two by two matrix-valued sequences,
and
\[
B_r(n,t)= \begin{pmatrix} B_{r,1}(n,t)\\ B_{r,2}(n,t) \end{pmatrix}
\]
is a vector in our Banach space with $B_{r,i}(n,t)=0$ for $|n|>\floor{\frac{r}{2}}+1$.
All entries of $A_{r,j}(t)$ and $B_r(t)$ are polynomials with respect to
$(a(n+j,t),b(n+j,t))$, $(a_\pm(n+j,t),b_\pm(n+j,t))$, $|j|\leq \floor{\frac{r}{2}}+1$.
Moreover, by our assumption the shift operators are continuous,
\[
\|S^\pm\| = \begin{cases} \sup_{n\in\Z} |\frac{w(n)}{w(n\pm1)}|^{1/p}, & p\in[1,\infty),\\
 \sup_{n\in\Z} |\frac{w(n)}{w(n\pm1)}|, & p=\infty, \end{cases}
\]
and same is true for the multiplication operators $A_{r,j}(t)$ whose norms depend only on the supremum
of the entries by H\"older's inequality, that is, on the sup norms of $(a(t),b(t))$ and $(a_\pm(t),b_\pm(t))$.
Finally, recall that by unitary equivalence of the operator family $H(t)$, respectively $H_\pm(t)$,
we have a uniform bound of the sup norm $\sup_n (|a(n,t)| + |b(n,t)|) \le 2\|H(t)\|= 2\|H(0)\|$, respectively
$\sup_n (|a_\pm(n,t)| + |b_\pm(n,t)|) \le 2\|H_\pm(t)\|= 2\|H_\pm(0)\|$.
Consequently, there is a constant such that  $\sum_{|j|\le r+1} \|A_{r,j}(t)\| \|(S^+)^j\| \le C_r$.
Moreover, we will show below that the vector $B_r(t)$ has only finitely many nonzero entries
and thus $\|B_r(t)\|_{w,p} \le D_r$, where the constant again depends only on the sup norms
of $(a(t),b(t))$ and $(a_\pm(t),b_\pm(t))$.
Hence
\[
\|\delta(t)\|_{w,p} \le \|\delta(0)\|_{w,p} + \int_0^t \big(C_r \|\delta(s)\|_{w,p} + D_r\big)
\]
and Gronwall's inequality implies
\[
\|\delta(t)\|_{w,p} \le \|\delta(0)\|_{w,p} \E^{C_r t} + \frac{D_r}{C_r} ( \E^{C_r t} -1).
\]
It remains to show existence of the above differential equation. This will follow once we
show that $g_{r+1}(t) - \bar g_{r+1}(t)$ and $h_{r+1}(t) - \bar h_{r+1}(t)$ can be written as
a linear combination of shifts of $\delta$ with the coefficients depending only on 
$(a(t),b(t))$ and $(a_\pm(t),b_\pm(t))$. The fact that $(\bar a, \bar b)$ does
not solve $\tl_r$ only affects finitely many terms and gives rise to the inhomogeneous term $B_r(t)$
which is nonzero only for a finite number of terms.

To see that $g_{r+1}(t) - \bar g_{r+1}(t)$ and $h_{r+1}(t) - \bar h_{r+1}(t)$ can be written as
a linear combination of shifts of $\delta$ we can use induction on $r$. It suffices to consider the
homogenous case where $c_j=0$, $1\leq j \leq r$, since all involved sums are finite. 
In this case \cite[Lemma~6.4]{tjac} shows that $g_j(n,t)$, $h_j(n,t)$ can be recursively computed from 
$g_0(n,t)=1$, $h_0(n,t)=0$ via
\begin{align} \nn
g_{j+1}(n,t) &= \frac{1}{2} \big(h_j(n,t) + h_j(n-1,t)\big)
+ b(n,t) g_j(n,t), \\ \nonumber
h_{j+1}(n,t)&= 2 a(n,t)^2 \sum_{l=0}^j g_{j-l}(n,t)g_l(n+1,t)
- \frac{1}{2} \sum_{l=0}^j h_{j-l}(n,t)h_l(n,t)
\end{align}
and similarly for $\bar g_j(n,t)$, $\bar h_j(n,t)$. Hence the claim follows.

Finally, observe that since $w(n)\ge 1$ this solution is bounded and hence
coincides with the solution of the Toda equation from Theorem~\ref{thmexistandunique}.
\end{proof}

For closely related results we also refer to \cite{ttd}.

\section{Quasi-periodic finite-gap solutions}
\label{secqp}

As a preparation for our next section we first need to recall some facts on
quasi-periodic finite-gap solutions (again see \cite{bght}, \cite{ghmt}, or \cite{tjac}).

Let $H_q^\pm$ be two quasi-periodic finite-band Jacobi operators,\footnote{Everywhere in this
paper the sub or super index "$+$" (resp.\ "$-$") refers to the background on the right
(resp.\ left) half-axis.}
\be
H_q^\pm(t) f(n) = a_q^\pm(n,t) f(n+1) + a_q^\pm(n-1,t) f(n-1) + b_q^\pm(n,t) f(n)
\ee
in $\ell^2(\Z)$ associated with the Riemann surface of the square root
\begin{equation}\label{defP}
P_\pm(z)= -\prod_{j=0}^{2g_\pm+1} \sqrt{z-E_j^\pm}, \qquad
E_0^\pm < E_1^\pm < \cdots < E_{2g_\pm+1}^\pm,
\end{equation}
where $g_\pm\in \N$ and $\sqrt{.}$ is the standard root with branch
cut along $(-\infty,0)$. In fact, $H_q^\pm(t)$ are uniquely determined by
fixing a Dirichlet divisor
$\sum_{j=1}^{g^\pm}(\mu_j^\pm(t), \sig_j^\pm(t))$, where
$\mu_j^\pm(t)\in[E_{2j-1}^\pm,E_{2j}^\pm]$ and $\sig_j^\pm(t)\in\{-1, 1\}$.
The time evolution of the Dirichlet divisor is determined by the Dubrovin
equations (cf.\ \cite[(13.2)]{tjac})and linearized by the Abel map (cf.\ \cite[Sect.~13.2]{tjac}).
The spectra of $H_q^\pm(t)$ consist of $g_\pm+1$ bands
\be \label{0.1}
\sig_\pm:= \sig(H_q^\pm(t)) = \bigcup_{j=0}^{g_\pm}
[E_{2j}^\pm,E_{2j+1}^\pm].
\ee
We will identify the set
$\C\setminus\sig(H_q^\pm(t))$ with the upper sheet of the Riemann surface.
Associated with $H_q^\pm(t)$ are the Weyl solutions
\be
\psi_q^\pm(z,n,t) \in \ell^2(\pm\N)
\ee
normalized such that $\psi_q^\pm(z,0,t)=1$. We will use the convention
that for $\lam\in\sig_\pm$ we set $\psi_q^\pm(\lam,n,t) = \lim_{\epsilon \downarrow 0}\psi_q^\pm(\lam + \I\epsilon)$.
Then
\be
\hat{\psi}_q^\pm(z,n,t) = \exp\big(\alpha_r^\pm(z,t)\big) \psi_q^\pm(z,n,t)
\ee
satisfies
\begin{align}
H_q^\pm(t) \hat\psi_q^\pm(z,n,t) &= z \hat\psi_q^\pm(z,n,t), \\
\frac{d}{dt} \hat\psi_q^\pm(z,n,t) &= P_{q,2r+2}^\pm(t) \hat\psi_q^\pm(z,n,t),
\end{align}
where (\cite{tjac}, (13.47))
\be
\alpha_r^\pm(z,t) =
\int_0^t \big(2 a_q^\pm(0,s) G_{q,r}^\pm(z,0,s) \psi_q^\pm(z,1,s) 
- H_{q,r+1}^\pm(z,0,s) \big)ds.
\ee
Note that the integrand in this last expression might have poles if $z$ lies in
one of the spectral gaps $[E^\pm_{2j-1},E^\pm_{2j}]$. Hence one has to understand
$\alpha_r^\pm(z,t)$ as a limit from $z\in\C\backslash\R$ for such values of $z$.
Alternatively one can use the expression in terms of Riemann theta functions.
Moreover, $\exp\big(\alpha_r^\pm(z,t)\big)$ has simple poles at $\mu^\pm_j(0)$
and simple zeros at  $\mu^\pm_j(t)$. We refer to the discussion in \cite{emtist} for
further details.

\section{Inverse scattering transform}
\label{secist}

Fix two quasi-periodic finite-gap solutions $a_q^\pm(n,t)$, $b_q^\pm(n,t)$ as in the previous section.
Let $a(n,t)$, $b(n,t)$ be a solution of the Toda hierarchy satisfying
\be \label{hypo}
\sum_{n = 0}^{\pm \infty} (1+|n|) \Big(|a(n,t) -
a_q^\pm(n,t)| + |b(n,t) - b_q^\pm(n,t)| \Big) < \infty
\ee
for one (and hence for any) $t_0 \in \R$. 
In \cite{emtstp2} (see also \cite{emtqps}, \cite{emtstp}, \cite{voyu}) we have developed scattering theory
for the Jacobi operator $H(t)$ associated with $a(n,t)$, $b(n,t)$. 
Jost solutions, transmission and reflection coefficients now depend on an 
additional parameter $t \in \R$. 
The essential spectrum of $H(t)$ is (absolutely) continuous and
\be
\sigma(H(t)) \equiv \sigma(H), \quad 
\sigma_{ess}(H)=\sig_+\cup\sig_-, \quad 
\sigma_p(H)=\{\lam_k\}_{k=1}^p \subseteq \R \backslash \sigma_{ess}(H),
\ee
where $p \in \N$ is finite. 
We introduce the sets
\be
\sig^{(2)}:=\sig_+\cap\sig_-, \quad
\sig_\pm^{(1)}=\clos\,(\sig_\pm\setminus\sig^{(2)}), \quad
\sig:=\sig_+\cup\sig_-,
\ee
where $\sig$ is the (absolutely) continuous
spectrum of $H(t)$ and $\sig_+^{(1)}\cup \sig_-^{(1)}$,
$\sig^{(2)}$ are the parts which are of multiplicity one, two, respectively.

The Jost solutions $\psi_\pm(z,n,t)$ are
normalized such that 
\be
\psi_\pm(z,n,t)= \psi_q^\pm(z,n,t)\,(1 + o(1)) \quad
\mbox{as } n\to\pm\infty.
\ee
Transmission $T_\pm(\lambda,t)$ and reflection $R_\pm(\lambda,t)$ 
coefficients are defined via the scattering relations
\be
T_\mp(\lambda,t) \psi_\pm(\lambda,n,t)= \overline{\psi_\mp(\lambda,n,t)} + 
R_{\mp}(\lambda,t) \psi_\mp(\lambda,n,t), \qquad \lambda \in \sigma_\mp,
\ee
which implies
\be
T_\pm(\lam,t):= \frac{W(\overline{\psi_\pm(\lam,t)}, \psi_\pm(\lam,t))}{W(\psi_\mp(\lam,t), \psi_\pm(\lam,t))},\quad
R_\pm(\lam,t):= - \frac{W(\psi_\mp(\lam,t),\overline{\psi_\pm(\lam,t)})}{W(\psi_\mp(\lam,t),\,\psi_\pm(\lam,t))},
\ee
$\lam\in \sigma_\pm$. 
Here $W_n(f,g)=a(n)(f(n)g(n+1)-f(n+1)g(n))$ denotes the usual Wronski determinant.

To define the norming constants we need to remove the poles of
$\psi^\pm(z,n,t)$ by introducing
\be
\tilde\psi^\pm(z,n,t)=\delta_\pm(z,t) \psi^\pm(z,n,t),\qquad
\delta_\pm(z,t) := \prod_{\mu^\pm_j(t)\in M_\pm(t)}  (z - \mu^\pm_j(t)),
\ee
where
\be
M^\pm(t) = \{\mu_j^\pm(t) \,|\, \mu_j^\pm(t) \in \R\backslash\sig_\pm \mbox{
is a pole of } \psi_q^\pm(z,1,t)\}.
\ee
The norming constants $\gamma_{\pm, k}(t)$ corresponding to $\lam_k \in \sigma_p(H)$
are then given by
\be
\gamma_{\pm, k}(t)^{-1}=\sum_{n \in \Z} |\tilde\psi_\pm(\lam_k, n,t)|^2.
\ee

\begin{lemma} \label{th: sol H t}
Let $(a(t),b(t))$ be a solution of the Toda hierarchy such that \eqref{hypo} holds.
The functions 
\be\label{hatpsi}
\hat\psi_\pm(z,n,t)=\exp(\alpha_r^\pm(z,t)) \psi_{\pm}(z,n,t) 
\ee
satisfy
\be \label{H t 1}
H(t) \hat\psi_\pm(z,n,t) = z \hat\psi_\pm(z,n,t), \qquad
\frac{d}{dt} \hat\psi_\pm(z,n,t) = P_{2r+2}(t) \hat\psi_\pm(z,n,t).
\ee
\end{lemma}

\begin{proof}
We proceed as in \cite[Theorem 3.2]{tist}. 
The Jost solutions $\psi_{\pm}(z,n,t)$ are continuously differentiable with respect to
$t$ by the same arguments as for $z$ (compare \cite[Theorem 4.2]{emtqps}),
and the derivatives are equal to the derivatives of the Baker--Akhiezer functions
as $n \rightarrow \pm \infty$.

For $z \in\C\backslash\sig$, the solution $u_{\pm}(z,n,t)$ of \eqref{H t 1} with initial condition 
$\psi_{\pm}(z,n,0) \in \ell_{\pm}^2(\Z)$ remains square summable near $\pm \infty$
for all $t \in \R$ (see \cite{ttkm} or \cite[Lemma 12.16]{tjac}), that is,
$u_{\pm}(z,n,t) = C_\pm(z,t) \psi_{\pm}(z,n,t)$. Letting $n\to\pm\infty$ we
see $C_\pm(z,t)=1$. The general result for all $z \in \C$ now follows from continuity.
\end{proof}

This implies

\begin{theorem} \label{thmtdscat}
Let $(a(t),b(t))$ be a solution of the Toda hierarchy such that \eqref{hypo} holds.
The time evolution for the scattering data is given by
\begin{align}
\begin{split}
T_\pm(\lambda, t) &= T_\pm(\lambda, 0)
\exp(\alpha_r^\mp(\lambda,t)- \ol{\alpha_r^\pm(\lambda,t)}), \\ 
R_\pm(\lambda, t) &= R_\pm(\lambda, 0) 
\exp(\alpha_r^\pm(\lambda,t)- \ol{\alpha_r^\pm(\lambda,t)}), \\  
\gamma_{\pm, k}(t) &= \gamma_{\pm, k}(0) \frac{\delta_\pm^2(\lam_k,0)}{\delta_\pm^2(\lam_k,t)} \exp(2 \alpha_r^\pm(\lam_k,t)), 
\qquad 1 \leq k \leq p.
\end{split}
\end{align}
\end{theorem}

\begin{proof}
The Wronskian of two solutions satisfying \eqref{H t 1} 
does not depend on $n$ or $t$ (see \cite{ttkm}, \cite[Lemma 12.15]{tjac}), hence
\begin{align*} 
T_\pm(\lambda, t) &= \frac{W(\overline{\psi_\pm(\lam,t)}, \psi_\pm(\lam,t))}{W(\psi_\mp(\lam,t), \psi_\pm(\lam,t))} 
= \frac{\exp(\alpha_r^\mp(\lambda,t))}
{\exp(\overline{\alpha_r^\pm(\lambda,t)})}
\frac{W(\overline{\hat \psi_\pm(\lam,t)}, \hat \psi_\pm(\lam,t))}
{W(\hat \psi_\mp(\lam,t), \hat \psi_\pm(\lam,t))}\\ 
&= \exp(\alpha_r^\mp(\lambda,t)-\overline{\alpha_r^\pm(\lambda,t)}) T_\pm(\lambda, 0), \\
R_{\pm}(\lambda, t) &= - \frac{W(\psi_\mp(\lam),\overline{\psi_\pm(\lam)})}{W(\psi_\mp(\lam), \psi_\pm(\lam))}=
- \frac{\exp(\alpha_r^\pm(\lambda,t))}
{\exp(\overline{\alpha_r^\pm(\lambda,t)})}
\frac{W(\hat \psi_\mp(\lam),\overline{\hat \psi_\pm(\lam)})}{W(\hat \psi_\mp(\lam), \hat \psi_\pm(\lam))} \\ \nn
&= \exp(\alpha_r^\pm(\lambda,t)-\overline{\alpha_r^\pm(\lambda,t)}) R_{\pm}(\lambda, 0).
\end{align*}
The time dependence of $\gamma_{\pm, k}(t)$ follows from
$\|U_r(t,0) \tilde\psi_\pm(\lam_k,.,0)\|= \|\tilde\psi_\pm(\lam_k,.,0)\|$.
\end{proof}

\begin{remark}
Note that we have
\be
\exp \big(\alpha_r^\pm(\lam,t) - \ol{\alpha_r^\pm(\lam,t)}\big) =
\exp\left( P_\pm(\lam) \int_0^t \frac{G_{q,r}^\pm(\lam,0,s)}{\prod_{j=1}^{g_\pm} (\lam-\mu^\pm_j(s))} ds\right),
\quad \lam\in\sig_\pm,
\ee
where $\mu^\pm_j(t)$ are the Dirichlet eigenvalues.
Moreover, in case of the Toda lattice, where $r=0$, we have $G_{q,0}^\pm(\lam,n,t)=1$ and
$H_{q,1}^\pm(\lam,n,t)=\lam-b^\pm_q(n,t)$.
\end{remark}

In summary, since Lemma~\ref{lemtsr} ensures that \eqref{hypo} remains
valid for all $t$ once it holds for the initial condition, we can compute
$R_\pm(\lambda,0)$ and $\gamma_{\pm,k}(0)$ from $(a(n,0), \linebreak b(n,0))$ and then
solve the Gel'fand--Levitan--Marchenko (GLM) equation to obtain the sequences
$(a(n,t), b(n,t))$ as in \cite{emtstp2}. More precisely, one needs to solve the GLM
equation
\be
K_\pm(n,m,t) + \sum_{l=n}^{\pm \infty} K_\pm(n,l,t)F_\pm(l,m,t) =
\frac{\delta_n(m)}{K_\pm(n,n,t)}, \qquad \pm m \geq \pm n,
\ee
for $K_\pm(n,m,t)$, where according to Theorem~\ref{thmtdscat} the kernel
$F_\pm(m,n,t)$ is given by

\begin{theorem}
The time dependence of the kernel of the Gel'fand--Levitan--\linebreak Marchenko equation is
given by
\begin{align} \nn
F_\pm(m,n,t) &=
\frac{1}{\pi} \re\int_{\sig_\pm}
R_\pm(\lambda,0) \hat \psi_q^\pm(\lambda, m,t) \hat \psi_q^\pm(\lambda, n, t) 
\frac{\prod_{j=1}^{g_\pm} (\lambda - \mu_j^\pm(0))}{P_\pm(\lambda)}  d\lambda\\
&\quad + \frac{1}{2\pi\I}\int_{\sig^{(1)}_\mp} |T_\mp(\lambda,0)|^2
\hat \psi_q^\pm(\lambda,m,t) \hat \psi_q^\pm(\lambda,n,t) \frac{\prod_{j=1}^{g_\pm} (\lambda - \mu_j^\pm(0))}{P_\mp(\lambda)}
d\lambda\\\nn
&\quad  + \sum_{k=1}^p \gamma_{\pm,k}(0) \breve \psi_q^\pm(\lam_k,m,t)\breve \psi_q^\pm(\lam_k, n,t), 
\end{align}
where $\breve\psi_q^\pm(z,m,t)=\delta_\pm(z,0)\hat\psi_q^\pm(z,m,t)$.
\end{theorem}

\begin{proof}
The kernel $F_\pm(m,n,0)$ is derived in \cite[Theorem 4.1]{emtstp2}.
Observe that $\alpha_r^\mp(\lambda,t)$ are real valued
on the set $\sig^{(1)}_\mp$ and 
\[
\exp \big(\alpha_r^\pm(\lam,t) + \ol{\alpha_r^\pm(\lam,t)}\big) 
= \prod_{j=1}^{g_\pm} \frac{\lambda - \mu_j^\pm(t)}{\lambda - \mu_j^\pm(0)}, \quad \lam \in \sig_\pm,
\]
then our result follows from \eqref{hatpsi} and Theorem~\ref{thmtdscat}.
\end{proof}

By \cite{emtstp2} this equation is uniquely solvable and the solution of the Toda hierarchy can
be obtained from either $K_+(n,m,t)$ or $K_-(n,m,t)$ by virtue of
\begin{align}\nn
a(n,t) &= a_q^+(n,t)\frac{K_+(n+1, n+1,t)}{K_+(n,n,t)} = a_q^-(n,t) \frac{K_-(n, n,t)}{K_-(n+1,n+1,t)}, \\ \nn
b(n,t) &= b_q^+(n,t) + a_q^+(n,t)\frac{K_+(n, n+1,t)}{K_+(n,n,t)} -
a_q^+(n-1,t) \frac{K_+(n-1,n,t)}{K_+(n-1,n-1,t)}, \\
 &= b_q^-(n,t) + a_q^-(n-1,t)\frac{K_-(n, n-1,t)}{K_-(n,n,t)} - a_q^-(n,t)
\frac{K_-(n+1, n,t)}{K_-(n+1,n+1,t)}.
\end{align}

%


\begin{thebibliography}{XXXX}
\bibitem{bdme2} A. Boutet de Monvel and I. Egorova, {\em  The Toda lattice with step-like
initial data. Soliton asymptotics}, Inverse Problems {\bf 16}, No. 4, 955--977 (2000).
\bibitem{bght} W. Bulla, F. Gesztesy, H. Holden, and G. Teschl, {\em
Algebro-Geometric Quasi-Periodic Finite-Gap Solutions of the Toda and Kac-van
Moerbeke Hierarchies}, Mem. Amer. Math. Soc. {\bf 135-641} (1998).
\bibitem{dkv} P. Deift, T. Kriecherbauer, and S. Venakides, {\em Forced lattice vibrations. I, II},
Comm. Pure Appl. Math. {\bf 48:11}, 1187--1249, 1251--1298 (1995).
\bibitem{dkkz} P. Deift, S. Kamvissis, T. Kriecherbauer, and X. Zhou,
{\em The Toda rarefaction problem}, Comm. Pure Appl. Math. {\bf 49}, no. 1, 35--83 (1996).
\bibitem{dei} K. Deimling, {\em Ordinary Differential Equations on Banach Spaces}, 
Lecture Notes in Mathematics {\bf 596}, Springer, Berlin, 1977.
\bibitem{egt} I. Egorova, K. Grunert, and G. Teschl,  {\em On the Cauchy problem for the
Korteweg-de Vries equation with steplike finite-gap initial data I. Schwartz-type perturbations},
Nonlinearity {\bf 22}, 1431--1457 (2009).
\bibitem{emtqps} I. Egorova, J. Michor, and G. Teschl, {\em Scattering theory for Jacobi
operators with quasi-periodic background}, Comm. Math. Phys. {\bf 264-3}, 811-842 (2006).
\bibitem{emtist} I. Egorova, J. Michor, and G. Teschl, {\em Inverse scattering transform for the Toda hierarchy
with quasi-periodic background}, Proc. Amer. Math. Soc. {\bf 135}, 1817--1827 (2007).
\bibitem{emtstp} I. Egorova, J. Michor, and G. Teschl, {\em Scattering theory for Jacobi operators
with steplike quasi-periodic background}, Inverse Problems {\bf 23}, 905--918 (2007).
\bibitem{emtstp2} I. Egorova, J. Michor, and G. Teschl, {\em Scattering theory for Jacobi operators
with general steplike quasi-periodic background}, Zh. Mat. Fiz. Anal. Geom. {\bf 4:1}, 33--62 (2008).
\bibitem{emt4} I. Egorova, J. Michor, and G. Teschl, {\em Soliton solutions of the Toda hierarchy on
quasi-periodic background revisited}, Math. Nach. {\bf 282:4}, 526--539 (2009).
\bibitem{fad} L. Faddeev and L. Takhtajan, {\em Hamiltonian Methods in the
Theory of Solitons}, Springer, Berlin, 1987.
\bibitem{ghmt} F. Gesztesy, H. Holden, J. Michor, and G. Teschl, {\em Soliton Equations
and Their Algebro-Geometric Solutions. Volume II: $(1+1)$-Dimensional Discrete Models},
Cambridge Studies in Advanced Mathematics {\bf 114}, Cambridge University Press, Cambridge, 2008.
\bibitem{gk} I. M. Guseinov and A. Kh. Khanmamedov, {\em The $t\to\infty$ asymptotics of the Cauchy problem
solution for the Toda chain with threshold-type initial data}, Th. and Math. Phys. {\bf 119:3}, 739--749 (1999).
\bibitem{hfm} B. L. Holian, H. Flaschka, and D. W. McLaughlin, {\em Shock waves in the Toda lattice: Analysis},
Ph. Rev. A {\bf 24:5}, 2595--2623 (1981).
\bibitem{km} S. Kamvissis, {\em On the long time behavior of the doubly infinite Toda
lattice under initial data decaying at infinity,}
Comm. Math. Phys., {\bf 153-3}, 479--519 (1993).
\bibitem{km2} S. Kamvissis, {\em On the Toda shock problem}, Physica D {\bf 65}, 242--266 (1993).
\bibitem{kt} S. Kamvissis and G. Teschl, {\em Stability of periodic soliton equations under short range perturbations},
Phys. Lett. A, {\bf 364-6}, 480--483 (2007).
\bibitem{kt2} S. Kamvissis and G. Teschl, {\em Stability of the periodic Toda lattice under short-range perturbations}, \arxiv{0705.0346}.
\bibitem{kt3} S. Kamvissis and G. Teschl, {\em Stability of the periodic Toda lattice: Higher order asymptotics}, \arxiv{0805.3847}.
\bibitem{kha2} A. Kh. Khanmamedov, {\em The solution of Cauchy's problem for the Toda lattice with limit periodic
initial data}, Sb. Math. {\bf 199:3}, 449--458 (2008).
\bibitem{krtrhp} H. Kr\"uger and G. Teschl, {\em Long-time asymptotics of the Toda lattice for decaying initial data revisited}, Rev. Math. Phys. {\bf 21:1}, 61--109 (2009).
\bibitem{krt} H. Kr\"uger and G. Teschl, {\em Long-time asymptotics for the Toda lattice in the soliton region}, Math. Z. {\bf 262}, 585--602 (2009).
\bibitem{krt2} H. Kr\"uger and G. Teschl, {\em Long-time asymptotics for the periodic Toda lattice in the soliton region}, Int. Math. Res. Not.
{\bf 2009}, Art. ID rnp077, 36 pp (2009).
\bibitem{lax} P. D. Lax, {\em Integrals of nonlinear equations of evolution and
solitary waves}, Comm. Pure and Appl. Math. {\bf 21}, 467--490 (1968).
\bibitem{mar} V. A. Marchenko, {\em Sturm--Liouville Operators and Applications}, 
Birkh\"auser, Basel, 1986.
\bibitem{mtqptr} J. Michor and G. Teschl, {\em Trace formulas for Jacobi operators
in connection with scattering theory for quasi-periodic background}, 
in Operator Theory, Analysis and Mathematical Physics, J. Janas (ed.) et al., 69--76,
Oper. Theory Adv. Appl. {\bf 174}, Birkh\"auser, Basel, 2007 
\bibitem{mtlax} J. Michor and G. Teschl, {\em On the equivalence of different Lax pairs for the Kac-van Moerbeke hierarchy},
in Modern Analysis and Applications, V. Adamyan (ed.) et al., 445--453, Oper. Theory Adv. Appl. {\bf 191}, Birkh\"auser, Basel, 2009.
\bibitem{nh} V. Yu. Novokshenov and I. T. Habibullin, {\em Nonlinear differential-difference
schemes integrable by the method of the inverse scattering problem. Asymptotics of the
solution for $t\to\infty$}, Sov. Math. Doklady {\bf 23/2}, 304--307 (1981).
\bibitem{oba} R. Oba, {\em Doubly Infinite Toda Lattice with Antisymmetric Asymptotics}, Ph.D. Thesis, New York 
University, 1988. 
\bibitem{tist} G. Teschl, {\em Inverse scattering transform for the Toda hierarchy},
Math. Nach. {\bf 202}, 163--171 (1999).
\bibitem{ttkm} G. Teschl, {\em On the Toda and Kac-van Moerbeke hierarchies},
Math. Z. {\bf 231}, 325--344 (1999).
\bibitem{tjac} G. Teschl, {\em Jacobi Operators and Completely Integrable Nonlinear Lattices},
Math. Surv. and Mon. {\bf 72}, Amer. Math. Soc., Rhode Island, 2000.
\bibitem{tag} G. Teschl, {\em Algebro-geometric constraints on solitons with respect to quasi-periodic backgrounds}, Bull. London Math. Soc.
{\bf 39:4}, 677--684 (2007).
\bibitem{ttd} G. Teschl, {\em On the spatial asymptotics of solutions of the Toda lattice}, \arxiv{0901.2717}.
\bibitem{ta} M. Toda, {\em Theory of Nonlinear Lattices}, 2nd enl. ed.,
Springer, Berlin, 1989.
\bibitem{vdo} S. Venakides, P. Deift, and R. Oba, {\em The Toda shock problem}, Comm. in Pure and Applied Math. {\bf 44}, 1171--1242 (1991).
\bibitem{voyu} A. Volberg and P. Yuditskii, {\em On the inverse scattering problem for Jacobi
Matrices with the Spectrum on an Interval, a finite systems of intervals or a Cantor set of positive
length}, Commun. Math. Phys. {\bf 226}, 567--605 (2002).
\end{thebibliography}
\end{document}